\documentclass[conference]{IEEEtran}

\usepackage{amsmath, amsthm, amssymb}
\usepackage{graphicx}
\usepackage{verbatim}
\usepackage{mathrsfs}
\usepackage{algorithmic}
\usepackage{float}
\usepackage[lofdepth, lotdepth]{subfig}
\usepackage{url}
\usepackage{array}

\usepackage[usenames,dvipsnames]{pstricks}
\usepackage{epsfig}
\usepackage{pst-grad} 
\usepackage{pst-plot} 

\theoremstyle{plain}
\newtheorem{theorem}{Theorem}
\newtheorem{lemma}[theorem]{Lemma}

\theoremstyle{definition}

\newcommand{\G}{{\mathcal G}}

\newcommand{\bM}{{\boldsymbol{M}}} 
\newcommand{\bN}{{\boldsymbol{N}}}

\newcommand{\bP}{{\boldsymbol P}}

\newcommand{\bu}{{\boldsymbol u}}

\newcommand{\bx}{{\boldsymbol{x}}}

\newcommand{\bG}{{\boldsymbol{G}}}

\newcommand{\fq}{\mathbb{F}_q}

\newcommand{\vep}{\varepsilon}
\newcommand{\de}{\delta}

\newcommand{\De}{\Delta}

\newcommand{\supp}{{\sf supp}}

\newcommand{\cG}{{\mathscr{G}}}
\newcommand{\cV}{{\mathscr{V}}}
\newcommand{\cE}{{\mathscr{E}}}
\newcommand{\cH}{{\mathscr{H}}}

\newcommand{\nin}{\noindent}

\newcommand{\seq}{\subseteq}

\newcommand{\vM}{{{\sf var}}(\boldsymbol{M})}
\newcommand{\gM}{{{\sf gr}}(\boldsymbol{M})}
\newcommand{\sG}{{{\sf supp}}(\boldsymbol{G})}

\newcommand{\jmax}{j_{\max}}
\newcommand{\jmin}{j_{\min}}

\newcommand{\bMi}{\boldsymbol{M}^{(i)}}

\newcommand{\Ri}{R^{(i)}}

\newcommand{\Ro}{R^{(1)}}
\newcommand{\Rt}{R^{(2)}}

\begin{document}
\pagestyle{plain}

\title{Balanced Sparsest Generator Matrices for MDS Codes}

\author{
   \IEEEauthorblockN{
     Son Hoang Dau\IEEEauthorrefmark{1},
     Wentu Song\IEEEauthorrefmark{2},
     Zheng Dong\IEEEauthorrefmark{3}, 
     Chau Yuen\IEEEauthorrefmark{4}}
   \IEEEauthorblockA{
     Singapore University of Technology and Design, Singapore\\     
		Emails: $\{$\IEEEauthorrefmark{1}{\it sonhoang\_dau},     
		\IEEEauthorrefmark{2}{\it wentu\_song},
    \IEEEauthorrefmark{3}{\it dong\_zheng},
     \IEEEauthorrefmark{4}{\it yuenchau}$\}$@sutd.edu.sg
 }
}
\maketitle

\begin{abstract}
We show that given $n$ and $k$, for $q$ sufficiently large, there
always exists an $[n, k]_q$ MDS code that has a generator matrix $\bG$ satisfying
the following two conditions: 
\begin{enumerate}
\item[(C1)] Sparsest: each row of $\bG$ has Hamming weight $n - k + 1$; 
\item[(C2)] Balanced: Hamming weights of the columns of $\bG$ differ from each other by at most one.
\end{enumerate}
\end{abstract}

\section{Introduction}
\label{sec:intro}

We study the existence and provide a construction of a \emph{sparsest} and \emph{balanced} 
generator matrix of Maximum Distance Separable (MDS) codes. 
A generator matrix is the sparsest if it contains
the least number of nonzero entries among all generator matrices of the same MDS code. 
A generator matrix is balanced if every column contains approximately the same
number of nonzero entries. More specifically, we require that the number of nonzero
entries in each column differs from each other by at most one. 

Apart from being of theoretical interest, our study on balanced sparsest generator matrices for
MDS codes was motivated by its application in error correction for sensor networks.
Suppose $n$ sensors, $S_1$, \ldots, $S_n$, collectively measure $k$ conditions 
$x_1, \ldots, x_k$, such as temperature, pressure, light intensity, etc.
Let $\bx = (x_1, \ldots, x_k)$, where $x_i \in \fq$ for each $i = 1,\ldots,k$ ($\fq$ is a finite field of $q$ elements). 
These sensors transmit the information they collected to a base station, which is a data collector.  
Furthermore, each sensor performs 
some encoding on the information it has, before transmitting the information back to the 
base station in the following way. Let $\bG$ be an $k \times n$ generator matrix of 
an $[n, k, d]_q$ error-correcting code. Sensor $S_i$
transmits the scalar product of $\bx$ and column $i$ of $\bG$ to the base station. 
It is well known in classical coding theory that this coding scheme allows the base station
to retrieve $\bx$ when at most $\lfloor \frac{d-1}{2}\rfloor$ sensors transmit wrong
information. Moreover, the base station can also identify the malfunctioned sensors.
For each sensor $S_i$, only those conditions corresponding to nonzero entries of column 
$i$ of $G$ are involved into encoding. So it is sufficient for $S_i$ to measure only such conditions. Thus, if $G$ is sparse then in average, each sensor only needs
to measure a few among $k$ conditions in order to achieve the desired error correction capability. 
On top of that, if columns
of $\bG$ have approximately the same number of nonzero entries then the sensors are required to
measure approximately the same number of conditions. This balance guarantees an even distribution
of workload among sensors, which is an important criterion for sensor networks
where energy saving is a critical issue. 

In fact, any error-correcting code can be used in the aforementioned scheme for
sensor networks. We choose to study MDS codes first because their structure, especially
their weight distribution, is well studied (see, for instance~\cite[Ch. 11]{MW_S}).
Moreover, they have optimal error-correcting capability, given the length and the 
dimension. 
We prove that over a sufficiently large field, there always exists an MDS code that
has a balanced and sparsest generator matrix, which is ideally suitable for the above encoding
scheme for sensor networks. 

Necessary notations and definitions are provided
in Section~\ref{sec:pre}. We state and prove our main result in Section~\ref{sec:main}. 

\section{Preliminaries}
\label{sec:pre}
 
We denote by $\fq$ the finite field with $q$ elements. 
Let $[n]$ denote the set $\{1,2,\ldots,n\}$. 
The \emph{support} of a vector $\bu = (u_1, \ldots, u_n) \in \fq^n$ is defined by
$\supp(\bu) = \{i \in [n]:\ u_i \neq 0\}$.
The (Hamming) \emph{weight} of $\bu$ is $|\supp(\bu)|$. 
We can also define weight and support of a row or a column of a matrix over some finite
field, by regarding them as vectors over that field. 
Apart from Hamming weight, we also use other standard notions from coding theory such as minimum distance, linear $[n,k]_q$ and $[n,k,d]_q$ codes, MDS codes, and generator matrices (for instance, see \cite{MW_S}). 
 
For a matrix $\bG = (g_{i,j}) \in \fq^{k \times n}$, the \emph{support matrix} of $\bG$, 
denoted $\sf{supp}(\bG)$, is a $k \times n$ binary matrix $\bM = (m_{i,j})$ 
where $m_{i,j} = 0$ if $g_{i,j} = 0$ and $m_{i,j} = 1$ if $g_{i,j} \neq 0$.  
Let $\bM = (m_{i,j})$ be a $k \times n$ binary matrix. 
We denote by $\vM =(v_{i,j})$ 
the matrix obtained from $\bM$ by replacing every nonzero entry $m_{i,j}=1$ by
$\xi_{i,j}$, where $\xi_{i,j}$'s are indeterminates. More formally, 
$v_{i,j} = 0$ if $m_{i,j} = 0$ and $v_{i,j} = \xi_{i,j}$ if $m_{i,j} = 1$. 
We also denote by $\sf{gr}(\bM)$ the bipartite graph $\cG = (\cV, \cE)$
defined as follows. The vertex set $\cV$ can be partitioned into two parts, 
namely, the left part $L = \{\ell_1, \ldots, \ell_k\}$, and the right part 
$R = \{r_1, \ldots, r_n\}$. The edge set is
\[
\cE = \big\{ (\ell_i, r_j):\ i \in [k], \ j \in [n], \ m_{i,j} \neq 0 \big\}. 
\]
For any $k \times n$ matrix $\bN$, we define 
$f(\bN) = \prod_{\bP} \det(\bP)$,
where the product is taken over all $\binom{n}{k}$ submatrices $\bP$ of order $k$ of $\bN$.  

\section{Main Result}
\label{sec:main}

A sparsest generator matrix of an $[n, k]_q$ MDS code would have precisely $n - k + 1$ nonzero entries in every row. Moreover, if it is balanced, then each column contains either $\lfloor \frac{k(n-k+1)}{n}\rfloor$ or $\lceil \frac{k(n-k+1)}{n}\rceil$ nonzero entries.  
Hereafter, we often use $R_i$, $i \in [k]$, and $C_j$, $j \in [n]$, to denote
the supports of row $i$ and column $j$, respectively, of a 
$k \times n$ binary matrix $\bM$. Note that $R_i \seq [n]$ and 
$C_j \seq [k]$. 

\vskip 5pt 
\begin{lemma}
\label{lem:1}
Let $\bM = (m_{i,j})$ be a $k \times n$ binary matrix.
Suppose that each row of $\bM$ has weight $n - k + 1$. 
Then $\bM$ is the support matrix of a generator matrix of some $[n,k]_q$ MDS code over 
a sufficiently large field $\fq$ ($q > \binom{n-1}{k-1}$) if and only if $f(\vM) \not\equiv 0$.
\end{lemma}
\vspace{-3pt}
\begin{proof} 
Suppose $\bM = \sG$, where $\bG = (g_{i,j})$ is a generator matrix of some $[n,k]_q$ MDS code. 
Due to a well-known property of MDS codes (see \cite[p. 319]{MW_S}), every submatrix
of order $k$ of $\G$ has nonzero determinant. Therefore, $f(\bG) \neq 0$. 
Note that $f(\vM)$ can be regarded as a multivariable polynomial in 
$\fq[\ldots, \xi_{i,j}, \ldots]$. Moreover, since $\bM = \sG$, we deduce that
$f(\bG)$ can be obtained from $f(\vM)$ by substituting $\xi_{i,j}$ by $g_{i,j}$ 
for all $i,j$ where $g_{i,j} \neq 0$. As $f(\bG) \neq 0$, we conclude that $f(\vM) \not\equiv 0$.   

Now suppose that $f(\vM) \not\equiv 0$.
Note that each column of $\vM$ belongs to precisely $\binom{n-1}{k-1}$
submatrices of order $k$ of $\vM$. Hence the exponent of each $\xi_{i,j}$
in $f(\vM)$ is at most $\binom{n-1}{k-1}$. Since $f(\vM) \not\equiv 0$, 
by \cite[Lemma 4]{Ho2006}, if $q > \binom{n-1}{k-1}$ then there exist
$g_{i,j} \in \fq$ (for $i,j$ where $m_{i,j} = 1$) so that 
$f(\vM)(\ldots, g_{i,j} , \ldots) \neq 0$.  
Let $\bG = (g_{i,j})$ (for $i,j$ where $m_{i,j} = 0$ we set $g_{i,j} = 0$).
Since 
$
f(\bG) =  f(\vM)(\ldots, g_{i,j}, \ldots) \neq 0$,
again by \cite[p. 319]{MW_S}, 
we deduce that $\bG$ is a generator matrix of an $[n, k]_q$ MDS code. 
Therefore, each row of $\bG$ has weight at least $n - k + 1$, due to the 
Singleton Bound (see \cite[p. 33]{MW_S}). 
Since each row of $\bM$ also has weight $n - k + 1$, we deduce that 
$g_{i,j} \neq 0$ whenever $m_{i,j} = 1$. Therefore, $\bM = \sG$.  
\end{proof} 
\vskip 5pt 

\begin{lemma}
\label{lem:2}
Let $\bM = (m_{i,j})$ be a $k \times n$ binary matrix.
Then $f(\vM) \not\equiv 0$ if and only if every bipartite subgraph induced by 
the $k$ left-vertices and some $k$ right-vertices in $\gM$ has a perfect matching.  
\end{lemma}
\vspace{-3pt}
\begin{proof}
Let $\cG = \gM$. Each submatrix $\bP$ of order $k$ of $\vM$ corresponds to a bipartite 
subgraph $\cH_{\bP}$ induced by the $k$ left-vertices and some $k$ right-vertices in $\cG$. 
In the literature, $\bP$ is usually referred to as the \emph{Edmonds matrix} of $\cH_{\bP}$. 
It is well known (see \cite[p. 167]{Motwani1995}) that a bipartite graph has a perfect matching if and only if 
the deteminant of its Edmonds matrix is not identically zero. Hence the proof follows.  
\end{proof} 

\begin{lemma}
\label{lem:3}
Let $\bM = (m_{i,j})$ be a $k \times n$ binary matrix.
Then every bipartite subgraph induced by the $k$ left-vertices and some $k$ right-vertices in $\gM$
has a perfect matching if and only if 
\begin{equation}
\label{eq:1} 
\big|  \cup_{j \in J} C_j\big| \geq |J|, \text{  for every subset } J \subseteq [n], \ |J| \leq k.
\end{equation} 
\end{lemma}
\vspace{-3pt}
\begin{proof}
Let $\cG = \gM$. Each submatrix $\bP$ of order $k$ of $\vM$ corresponds to a bipartite 
subgraph $\cH_{\bP}$ induced by the $k$ left-vertices and some $k$ right-vertices in $\cG$. 
The lemma follows by applying Hall's marriage theorem to each of such subgraphs
of $\gM$.  
\end{proof} 

\begin{lemma}
\label{lem:4}
Let $\bM = (m_{i,j})$ be a $k \times n$ binary matrix.
The condition (\ref{eq:1}) is equivalent to 
\vspace{-3pt}
\begin{equation} 
\label{eq:2}
\big| \cup_{i \in I} R_i \big| \geq n - k + |I|, \text{  for every subset } \varnothing \neq I \subseteq [k]. 
\end{equation} 
\end{lemma}
\vspace{-5pt}
\begin{proof}
Suppose that (\ref{eq:1}) holds and that there exists a nonempty set $I \subseteq [k]$ satisfying
\begin{equation} 
\label{eq:3}
\big| \cup_{i \in I} R_i \big| \leq n - k + |I| - 1. 
\end{equation}
We aim to obtain a contradiction. 
The condition (\ref{eq:3}) is equivalent to
\begin{equation} 
\label{eq:4}
\big| \cap_{i \in I} \overline{R_i} \big| \geq k - |I| + 1.
\end{equation} 
Hence there exists a set $J$ of $k - |I| + 1$ columns of $\bM$ that satisfies 
$
\big|  \cap_{j \in J} \overline{C_j}\big| \geq |I|$.  
Equivalently we have
\begin{equation} 
\label{eq:5}
\big|  \cup_{j \in J} C_j\big| \leq k - |I| < k - |I| + 1 = |J|. 
\end{equation} 
We obtain a contradiction between (\ref{eq:1}) and (\ref{eq:5}). 
The ``only if" direction can be proved in a similar manner. 
\end{proof} 

\begin{lemma}
\label{lem:main1}
Let $\bM = (m_{i,j})$ be a $k \times n$ binary matrix.
Suppose that each row of $\bM$ has weight $n - k + 1$. 
Then $\bM$ is the support matrix of a generator matrix of some $[n,k]_q$ MDS code over 
a sufficiently large field $\fq$ ($q > \binom{n-1}{k-1}$) if and only if (\ref{eq:2}) holds. 
\end{lemma}
\begin{proof} 
The proof follows from Lemma~\ref{lem:1}-\ref{lem:4}.
\end{proof}
\vskip 5pt

We present below our main result. 
\vskip 5pt 
\begin{theorem}[Main Theorem]
\label{thm:main}
Suppose $1 \leq k \leq n$ and $q > \binom{n-1}{k-1}$. Then there always exists an 
$[n,k]_q$ MDS code that has a generator matrix $\bG$ satisfying the following two conditions.
\begin{enumerate}
	\item[(C1)] Sparsest: each row of $\bG$ has weight $n - k + 1$.  
	\item[(C2)] Balanced: column weights of $\bG$ differ from each other by at most one. 
\end{enumerate}
\end{theorem}

By Lemma~\ref{lem:main1}, to prove Theorem~\ref{thm:main}, we need to 
show that there always exists a $k \times n$ binary matrix $\bM$ satisfying the
following properties

\begin{enumerate}
	\item[(P1)] each row of $\bM$ has weight $n - k + 1$, 
	\item[(P2)] column weights of $\bM$ differ from each other by at most one,
	\item[(P3)] $\big| \cup_{i \in I} R_i \big| \geq n - k + |I|$, for every subset $\varnothing \neq I \subseteq [k]$, 
where $R_i$ denotes the support of row $i$ of $\bM$. 	
\end{enumerate}

We prove the existence of such a binary matrix by designing an algorithm 
(Algorithm 1) that starts from an initial
binary matrix which satisfies (P1) and (P3). In each iteration, the matrix at
hand is slightly modified so that it still satisfies (P1) and (P3) and its
column weights become more balanced. When the algorithm terminates, 
it produces a matrix that satisfies (P1), (P2), and (P3). 

Observe that it is fairly easy to construct a binary matrix that satisfies (P1) and (P2), 
using the Gale-Ryser Theorem (see Manfred~\cite{Manfred1996}).
However, (P1) and (P2) do not automatically guarantee (P3). 
Indeed, the matrix $\bP$ given below satisfies both (P1) and (P2). 
However, (P3) is violated if we choose $I = \{1,2,3\}$.  
\[
\bP = 
\begin{pmatrix}
1& 0& 0& 0& 1& 1 & 1 & 0 \\
1& 0& 0& 0& 1& 0 & 1 & 1 \\
1& 0& 0& 0& 0& 1 & 1 & 1 \\
0& 1& 1& 1& 1& 0 & 0 & 0 \\
0& 1& 1& 1& 0& 1 & 0 & 0 
\end{pmatrix}.
\]
Let $\widetilde{\bM}$ be any $k \times n$ binary matrix 
that satisfies both (P1) and (P3). 
For instance, we can shift the vector $(\underbrace{1\ 1\ \cdots \ 1}_{n-k+1} 0\ 0\ \cdots 0)$
$k$ times cyclically to produce $k$ rows of such a matrix as below. 
\[
\widetilde{\bM} = \begin{pmatrix}
1& 1& 1& \cdots & 1 & 0 & 0 & 0 & \cdots & 0\\
0& 1& 1& \cdots & 1 & 1 & 0 & 0 & \cdots & 0\\
0& 0& 1& \cdots & 1 & 1 & 1 & 0 & \cdots & 0\\
\vdots & \vdots & \vdots & \ddots & \vdots & \vdots & \vdots & \vdots & \ddots & \vdots\\
0& 0& 0 & \cdots & 1 & 1 & 1 & 1 & \cdots & 1\\
\end{pmatrix}.
\] 
The Algorithm~1 takes $\widetilde{\bM}$ as an input parameter.  
\begin{figure}[H]
\centering
\fbox{
\parbox{3.2in}{
\centerline{\bf Algorithm 1}
\nin{\bf Input:} $n$, $k$, $\widetilde{\bM}$;\\
\nin{\bf Initialization:} $\bM := \widetilde{\bM}$;  
\begin{algorithmic}[1]
\REPEAT
\STATE{Let $\max$ and $\min$ be the maximum and minimum weights of columns of $\bM$;}
\IF {$\max - \min \leq 1$}
   \STATE{Return $\bM$;}
\ENDIF
\STATE{Find two columns $j_{\max}$ and $j_{\min}$ that have weights $\max$ and $\min$, respectively;}
\STATE{Find a row $i_s$ satisfying $m_{i_s,j_{\max}} = 1$ and $m_{i_s,j_{\min}} = 0$ and moreover, 
if we set $m_{i_s,j_{\max}} := 0$ and $m_{i_s,j_{\min}} := 1$ then $\bM$ still satisfies (P1) and (P3);}
\STATE{Swapping: set $m_{i_s,j_{\max}} = 0$ and $m_{i_s,j_{\min}} := 1$;}
\UNTIL{$\max - \min \leq 1$;}
\end{algorithmic}
}
}
\end{figure}
\vspace{-10 pt}
Due to space constraint, we have prepared a separate note at \cite{supplement} with an example to demonstrate the algorithm. 
\begin{lemma}
\label{lem:correctness}
Suppose in every iteration, Algorithm~1 can always find a legitimate row described in Step~7. 
Then the algorithm terminates after finitely many iterations and returns a matrix satisfying (P1), (P2), and (P3). 
\end{lemma}
\vspace{-3pt}
\begin{proof}  
At a certain iteration, let $\De = \max - \min$. 
After swapping the two entries $m_{i_s,j_{\max}}$ and $m_{i_s,j_{\min}}$, 
the weight of column $j_{\max}$ is decreased by one whereas the weight
of column $j_{\min}$ is increased by one. Therefore, after at most
$\lfloor n/2 \rfloor$ iterations, $\De$ is decreased by at least one. 
Hence, the algorithm must terminate after finitely many iterations.
The ouput matrix obviously satisfies (P1), (P2), and (P3).    
\end{proof} 
\vspace{-3pt}
\begin{lemma}
\label{lem:main2}
In every iteration of Algorithm~1, a row $i_s$ as described in Step~7 
of the algorithm can always be found. 
\end{lemma}
\vspace{-3pt}
Since column $j_{\max}$ has a larger weight than column $j_{\min}$, there always
exists at least one row $i_s$ where $m_{i_s,j_{\max}} = 1$ and $m_{i_s,j_{\min}} = 0$.
Obviously, swapping $m_{i_s,j_{\max}}$ and $m_{i_s,j_{\min}}$ does not make 
$\bM$ violate (P1). 
The stricter criterion is that $\bM$ must still satisfy (P3) after the swap. 
We need a few more auxiliary results
before we can prove Lemma~\ref{lem:main2}. 

Suppose at a certain iteration, we choose some columns $\jmax$ and $\jmin$
that have maximum and minimum weights, respectively. Without loss of generality, 
we assume that the first $t$ rows are all the rows of $\bM$ satisfying 
the property that each of them has a one at column $\jmax$ and a zero at column $\jmin$. In other words, assume that
\[
\{i \in [k]: \ m_{i,j_{\max}} = 1 \text{ and } m_{i,j_{\min}} = 0\} = [t]. 
\]
Since $\max - \min \geq 2$, we have $t \geq 2$.
  
Suppose, for contradiction, that none of these $t$ rows satisfy the condition in 
Step~7 of Algorithm~1. Let $\bMi$, $i \in [t]$, be the matrix obtained from $\bM$ after
swapping the two entries $m_{i,\jmax}$ and $m_{i,\jmin}$. Then $\bMi$, $i \in [t]$, 
does not satisfy (P3). Since $\bM$ satisfies (P3) and the only difference between $\bMi$
and $\bM$ is the row $i$, the set of rows of $\bMi$ that violates the condition (P3) must 
contain row $i$. Therefore, for each $i \in [t]$, there exists a set $I_i \subset [k]$, 
$i \notin I_i$, such that $\{i\} \cup I_i$
is a set of rows that violates (P3) in $\bMi$. 
For our purpose, for each $i \in [t]$, we choose $I_i$ to be of minimum size among those sets that 
satisfied the aforementioned requirement.
Since for each $i\in[t]$, $|\Ri_i|=|R_i|=n-k+1$, we deduce that $I_i\neq\emptyset$. 

Let $R^{(i)}_r$ denote the support of row $r$ of $\bMi$, $i \in [t]$, $r \in [k]$. 
Note that $R_r$ denotes the support of row $r$ of $\bM$, $r \in [k]$. 
For simplicity, 
we use $\Ri_I$ to denote the union $\cup_{r \in I} \Ri_r$ for any subset $I \subseteq [k]$. 
Since $\{i\} \cup I_i$ is the set of rows of $\bMi$ that violates (P3),  for every $i \in [t]$ we have
\begin{equation}
\label{eq:6} 
|\Ri_{\{i\} \cup I_i}| \leq n - k + |\{i\} \cup I_i| - 1 = n - k + |I_i|. 
\end{equation} 
\begin{lemma}
\label{lem:5}
For all $i, i' \in [t]$, the following statements hold
\begin{enumerate}
	\item[a)] $\Ri_{i'} = \begin{cases} R_{i'}, & \text{ if } i' \neq i,\\ (R_{i'} \setminus \{\jmax\})
	\cup \{\jmin\}, & \text{ if } i' = i,\end{cases}$
	\item[b)] $\Ri_{I_i} = R_{I_i}$, \quad c) $\jmax \notin R_{I_i}$, \quad d) $\jmin \in R_{I_i}$, 
	\item[e)] $i \notin I_{i'}$, \quad f) $|\Ri_{\{i\} \cup I_i}| = n - k + |I_i|$.
\end{enumerate}
\end{lemma}
\begin{proof}
\nin {\bf Proof of a).}
Note that all the rows of $\bMi$ except for the row $i$ are the same as that of $\bM$.
Therefore, $\Ri_{i'} = R_{i'}$ if $i' \neq i$. 
As row $i$ of $\bMi$ is obtained from row $i$ of $\bM$ by swapping 
$m_{i,\jmax}=1$ and $m_{i,\jmin} = 0$, we deduce that 
\[
\Ri_{i} = (R_{i} \setminus \{\jmax\}) \cup \{\jmin\}. 
\] 
\nin {\bf Proof of b).} 
By definition of $I_i$, $i \notin I_i$. Therefore, using Part a), we conclude that 
$\Ri_{I_i} = R_{I_i}$.\\ 
\nin {\bf Proof of c).} 
Suppose, for contradiction, that $\jmax \in R_{I_i}$. 
Due to Part a) and b), we have
\[
\begin{split} 
\Ri_{\{i\} \cup I_i} &= \Ri_i \cup \Ri_{I_i}\\
&= ((R_i \setminus \{\jmax\}) \cup \{\jmin\}) \cup R_{I_i}\\
&= ((R_i \setminus \{\jmax\}) \cup R_{I_i}) \cup \{\jmin\}\\
&= (R_i \cup R_{I_i}) \cup \{\jmin\} \supseteq R_{\{i\} \cup I_i}. 
\end{split} 
\] 
As $\bM$ satisfies (P3), we have
\[
|\Ri_{\{i\} \cup I_i}| \geq |R_{\{i\} \cup I_i}| \geq n - k + |\{i\} \cup I_i| = n - k + |I_i| + 1. 
\]
This inequality contradicts (\ref{eq:6}).\\
\nin {\bf Proof of d).} 
Suppose, for contradiction, that $\jmin \notin R_{I_i}$. 
Then by Part a) and b) we have
\[
\Ri_{\{i\} \cup I_i} = ((R_i \setminus \{\jmax\}) \cup \{\jmin\}) \cup R_{I_i} \supseteq \{\jmin\} \cup R_{I_i}. 
\]
Therefore, using the fact that $\bM$ satisfies (P3), we deduce that
\[
|\Ri_{\{i\} \cup I_i}| \geq |\{\jmin\} \cup R_{I_i}| = 1 + |R_{I_i}| \geq 1 + n - k + |I_i|. 
\]
This inequality contradicts (\ref{eq:6}). \\
\nin {\bf Proof of e).} 
Note that $\jmax \in R_i$. However, by Part c), $\jmax \notin R_{I_{i'}}$. 
Hence, $i \notin I_{i'}$. \\
\nin {\bf Proof of f).} 
Using Part a) we have 
\begin{equation} 
\label{eq:7}
|\Ri_{\{i\} \cup I_i}| = |\Ri_i \cup \Ri_{I_i}| = |\Ri_i \cup R_{I_i}| \geq |R_{I_i}| 
\geq n - k + |I_i|,
\end{equation} 
where the last inequality comes from the fact that $\bM$ satisfies (P3). 
Combining (\ref{eq:6}) and (\ref{eq:7}), the proof of f) follows.  
\end{proof} 

\vskip 5pt 
\begin{lemma}
\label{lem:6}
For all $i, i' \in [t]$, $i \neq i'$, it holds that $I_i \cap I_{i'} = \varnothing$. 
\end{lemma}
\begin{proof}
Without loss of generality, we prove that $I_1 \cap I_2 = \varnothing$. 
Suppose, for contradiction, that there exists $\ell \in I_1 \cap I_2$. 
We first present three claims, which are used later in this proof. \\
\nin {\bf Claim 1:}
For $i = 1, 2$ we have
\begin{equation} 
\label{eq:8}
|\Ri_{\{i\} \cup I_i \setminus \{\ell\}}| = n - k + |I_i|,
\end{equation} 
and
\begin{equation} 
\label{eq:9}
R_{\ell} = \Ri_{\ell} \seq \Ri_{\{i\} \cup I_i \setminus \{\ell\}}. 
\end{equation} 
\begin{proof} [Proof of Claim 1]
Indeed, because of the minimality of $I_i$, the set $\{i\} \cup (I_i \setminus \ell)$
does not violate (P3) in $\bMi$. Therefore, 
\[
|\Ri_{\{i\} \cup I_i \setminus \{\ell\}}| \geq n - k + |\{i\}\cup (I_i \setminus \ell)|
= n - k + |I_i|. 
\]
On the other hand, 
$
\Ri_{\{i\} \cup I_i \setminus \{\ell\}} \seq \Ri_{\{i\} \cup I_i}$,
which also has cardinality $n - k + |I_i|$, due to Lemma~\ref{lem:5} f). 
Therefore, 
\[
\Ri_{\{i\} \cup I_i \setminus \{\ell\}} = \Ri_{\{i\} \cup I_i},
\]
and 
\[
|\Ri_{\{i\} \cup I_i \setminus \{\ell\}}| = |\Ri_{\{i\} \cup I_i}| = n - k + |I_i|.
\]
We also deduce that  
$
\Ri_{\ell} \seq \Ri_{\{i\} \cup I_i \setminus \{\ell\}}$. 
By Lemma~\ref{lem:5} a), we have $R_{\ell} = \Ri_{\ell}$. 
Thus we complete the proof of Claim~1. 
\end{proof} 

\nin {\bf Claim 2:} Let $K = (I_1\setminus \{\ell\}) \cap (I_2 \setminus \{\ell\})$. 
Then for $i = 1, 2$, the following holds 
\begin{equation} 
\label{eq:14}
|R_{\{i\} \cup I_i \setminus \{\ell\}} \setminus R_{\{\ell\} \cup K}  \setminus \{\jmax\}|
\leq |\Ri_{\{i\} \cup I_i \setminus \{\ell\}}| - |R_{\{\ell\} \cup K}|.  
\end{equation} 
\begin{proof}[Proof of Claim 2]
Using Lemma~\ref{lem:5} a) and b), we have
\[
\begin{split} 
R_{\{i\} \cup I_i \setminus \{\ell\}} &= R_i \cup R_{I_i \setminus \{\ell\}}\\
&= R_i \cup \Ri_{I_i \setminus \{\ell\}}\\
&= ((\Ri_i \cup \{\jmax\}) \setminus \{\jmin\}) \cup \Ri_{I_i \setminus \{\ell\}}\\
&\seq \{\jmax\} \cup  (\Ri_i \cup \Ri_{I_i \setminus \{\ell\}})\\
&= \{\jmax\} \cup \Ri_{\{i\} \cup I_i \setminus \{\ell\}}. 
\end{split} 
\]
Therefore, 
$
R_{\{i\} \cup I_i \setminus \{\ell\}} \setminus \{\jmax\} 
\seq \Ri_{\{i\} \cup I_i \setminus \{\ell\}}$. 
Hence 
\[
\begin{split}
|(R_{\{i\} \cup I_i \setminus \{\ell\}}  \setminus  R_{\{\ell\} \cup K}) \setminus  \{\jmax\}|
&= |(R_{\{i\} \cup I_i \setminus \{\ell\}}  \setminus \{\jmax\}) \\
&\quad \setminus  R_{\{\ell\} \cup K}|\\
&\leq |\Ri_{\{i\} \cup I_i \setminus \{\ell\}} \setminus R_{\{\ell\} \cup K}| \\
&=|\Ri_{\{i\} \cup I_i \setminus \{\ell\}}| - |R_{\{\ell\} \cup K}|,
\end{split} 
\]
where the last equality can be explained by the fact that $R_{\{\ell\} \cup K} 
\seq \Ri_{\{i\} \cup I_i \setminus \{\ell\}}$. 
Indeed, we have 
\[
R_{\ell} \seq \Ri_{\{i\} \cup I_i \setminus \{\ell\}}
\]
due to (\ref{eq:9}). Moreover, $K \seq I_i \setminus \{\ell\}$. 
Hence the aforementioned inclusion holds. 
We complete the proof of Claim~2. 
\end{proof} 

\nin {\bf Claim 3:} If $I_1 \cap I_2 = \{\ell\}$ then for $i = 1, 2$, we have
\begin{equation} 
\label{eq:10}
|R_{\{i\} \cup I_i \setminus \{\ell\}} \setminus R_{\ell} 
\setminus \{\jmax\}| \leq |I_i| - 1.  
\end{equation} 
\begin{proof}[Proof of Claim~3]
Applying (\ref{eq:14}) with $K = \varnothing$, we obtain
\[
\begin{split}
|R_{\{i\} \cup I_i \setminus \{\ell\}}  \setminus  R_{\ell} \setminus  \{\jmax\}|
&\leq |\Ri_{\{i\} \cup I_i \setminus \{\ell\}}| - |R_{\ell}|\\
&\stackrel{(\ref{eq:8})}{=} (n - k + |I_i|) - (n - k + 1)\\
&= |I_i| - 1. 
\end{split} 
\]
We complete the proof of Claim~3. 
\end{proof} 

The remaining of the proof of Lemma~\ref{lem:6} is divided into two cases. 
Our goal is to obtain contradictions in both cases.
\nin {\bf Case 1:} $I_1 \cap I_2 = \{\ell\}$. \\
We aim to show that 
\begin{equation} 
\label{eq:11}
|R_{\{1,2\}  \cup I_1 \cup I_2}| < n - k + |\{1,2\}  \cup I_1 \cup I_2|. 
\end{equation} 
This is a contradiction of our assumption that $\bM$ satisfies (P3). 
Firstly, since $I_1 \cap I_2 = \{\ell\}$, we have
\begin{equation} 
\label{eq:12}
n - k + |\{1,2\}  \cup I_1 \cup I_2| = n - k + |I_1| + |I_2| + 1. 
\end{equation} 
Secondly, we consider
\[
\begin{split} 
R_{\{1,2\}  \cup I_1 \cup I_2}
&= R_{\ell} \cup (R_{\{1\} \cup I_1 \setminus \{\ell\}} \setminus R_{\ell}) 
\cup (R_{\{2\} \cup I_2 \setminus \{\ell\}} \setminus R_{\ell})\\
&= R_{\ell} \cup \{\jmax\} 
\cup ((R_{\{1\} \cup I_1 \setminus \{\ell\}} \setminus R_{\ell}) \setminus \{\jmax\})\\ 
&\quad \cup ((R_{\{2\} \cup I_2 \setminus \{\ell\}} \setminus R_{\ell}) \setminus \{\jmax\}).
\end{split} 
\]
Therefore, 
\begin{equation} 
\label{eq:13}
\begin{split} 
|R_{\{1,2\}  \cup I_1 \cup I_2}|
&\leq |R_{\ell}| + 1 
+ |(R_{\{1\} \cup I_1 \setminus \{\ell\}} \setminus R_{\ell}) \setminus \{\jmax\}|\\ 
&\quad + |(R_{\{2\} \cup I_2 \setminus \{\ell\}} \setminus R_{\ell}) \setminus \{\jmax\}|\\
&\stackrel{(\ref{eq:10})}{\leq} (n - k + 1) + 1 + (|I_1| - 1) + (|I_2| - 1)\\
&= n - k + |I_1| + |I_2|.  
\end{split} 
\end{equation} 
Combining (\ref{eq:12}) and (\ref{eq:13}), we obtain (\ref{eq:11}). 
We complete the analysis of Case 1. 

\nin {\bf Case 2:} $(I_1\setminus \{\ell\}) \cap (I_2 \setminus \{\ell\}) = K \neq \varnothing$. \\
We aim to prove that
\begin{equation}
\label{eq:15}
R_{\ell} \seq R_K, 
\end{equation} 
and
\begin{equation}
\label{eq:16}
|R_K| = n - k + |K|. 
\end{equation} 
If both (\ref{eq:15}) and (\ref{eq:16}) hold then 
\[
|R_{\{\ell\} \cup K}| = |R_K| = n - k + |K| < n - k + |\{\ell\} \cup K|, 
\]
which contradicts our assumption that $\bM$ satisfies (P3). 

Let $\de = |R_{\ell} \setminus R_K| \geq 0$.
As $\bM$ satisfies (P3), let
\[ 
|R_K| = n - k + |K| + \vep,
\]
where $\vep \geq 0$. Then 
\begin{equation} 
\label{eq:17}
|R_{\{\ell\} \cup K}| = |R_K| + |R_{\ell} \setminus R_K| = n - k + |K| + \vep + \de. 
\end{equation} 
We have
\[
\begin{split}
R_{\{1,2\}  \cup I_1 \cup I_2}
&= R_{\{\ell\} \cup K} \cup (R_{\{1\} \cup I_1 \setminus \{\ell\}} 
\setminus R_{\{\ell\} \cup K})\\ 
&\quad \cup (R_{\{2\} \cup I_2 \setminus \{\ell\}} \setminus R_{\{\ell\} \cup K})\\
&=R_{\{\ell\} \cup K} \cup \{\jmax\} \\
&\quad \cup ((R_{\{1\} \cup I_1 \setminus \{\ell\}} \setminus R_{\{\ell\} \cup K}) 
\setminus \{\jmax\})\\ 
&\quad \cup ((R_{\{2\} \cup I_2 \setminus \{\ell\}} \setminus R_{\{\ell\} \cup K}) 
\setminus \{\jmax\}).
\end{split} 
\]
Therefore
\begin{equation} 
\label{eq:18}
\begin{split}
|R_{\{1,2\}  \cup I_1 \cup I_2}|
&\leq |R_{\{\ell\} \cup K} \cup \{\jmax\}| \\
&\quad + |(R_{\{1\} \cup I_1 \setminus \{\ell\}} \setminus R_{\{\ell\} \cup K}) 
\setminus \{\jmax\}|\\ 
&\quad + |(R_{\{2\} \cup I_2 \setminus \{\ell\}} \setminus R_{\{\ell\} \cup K}) 
\setminus \{\jmax\}|\\
&\stackrel{(\ref{eq:14})}{\leq} |R_{\{\ell\} \cup K}| + 1\\
&\quad + |\Ro_{\{1\} \cup I_1 \setminus \{\ell\}}| - |R_{\{\ell\} \cup K}|\\
&\quad + |\Rt_{\{2\} \cup I_2 \setminus \{\ell\}}| - |R_{\{\ell\} \cup K}|\\
&= |\Ro_{\{1\} \cup I_1 \setminus \{\ell\}}| 
+ |\Rt_{\{2\} \cup I_2 \setminus \{\ell\}}| \\
&\quad - |R_{\{\ell\} \cup K}| + 1\\
&\stackrel{(\ref{eq:8})(\ref{eq:17})}{=} 
(n - k + |I_1|) + (n - k + |I_2|) \\
&\quad - (n - k + |K| + \vep + \de) + 1\\
&\leq n - k + |I_1| + |I_2| - |K| + 1. 
\end{split} 
\end{equation} 
Moreover, as $I_1 \cap I_2 = \{\ell\} \cup K$, we have
\begin{equation} 
\label{eq:21}
\begin{split} 
|\{1,2\}\cup I_1 \cup I_2| &= 2 + |I_1| + |I_2| - |\{\ell\} \cup K|\\
&= |I_1| + |I_2| - |K| + 1. 
\end{split} 
\end{equation} 
As $\bM$ satisfies (P3), from (\ref{eq:18}) and (\ref{eq:21}), we conclude that
\[
|R_{\{1,2\}  \cup I_1 \cup I_2}| = n - k + |I_1| + |I_2| - |K| + 1.
\]
Therefore, all of the inequalities in (\ref{eq:18}) must be equalities. 
In particular, the last equality forces $\vep = 0$ and $\de = 0$. 
As $\de = 0$ implies that (\ref{eq:15}) holds and $\vep = 0$ implies
that (\ref{eq:16}) holds, we complete the analysis of Case~2.

In any cases, we always derive a contradiction. Therefore, our 
assumption that there exists some $\ell \in I_1 \cap I_2$ is wrong. 
Hence $I_1 \cap I_2 = \varnothing$. It follows immediately 
that $I_i \cap I_{i'} = \varnothing$ for every $i, i' \in [t]$, 
$i \neq i'$. 
\end{proof} 
\vskip 5pt 

We are now in position to prove Lemma~\ref{lem:main2}, 
which in turn implies Theorem~\ref{thm:main}. 

\vskip 5pt 
\begin{proof}[Proof of Lemma~\ref{lem:main2}]
Recall that we assume that
\begin{equation} 
\label{eq:19}
\{i \in [k]: \ m_{i,j_{\max}} = 1 \text{ and } m_{i,j_{\min}} = 0\} = [t]. 
\end{equation} 
Moreover, we suppose, for contradiction, that none of these $t$ rows 
satisfy the second condition in Step~7 of Algorithm~1.
As shown by Lemma~\ref{lem:5} c), d), e), and Lemma~\ref{lem:6}, 
we can associate to each $i \in [t]$ a subset $I_i \subset [k]$ satisfying
the following
\begin{enumerate}
	\item[(S1)] $i \notin I_{i'}$, for all $i, i' \in [t]$, 
	\item[(S2)] $\jmax \notin R_{I_i}$, for all $i \in [t]$, 
	\item[(S3)] $\jmin \in R_{I_i}$, for all $i \in [t]$, 
	\item[(S4)] $I_i \cap I_{i'} = \varnothing$, for all $i, i' \in [t]$, $i \neq i'$. 
\end{enumerate}
Due to (S2) and (S3), for each $i \in [t]$, there exists a row $r(i) \in I_i$
that has a zero at column $\jmax$ and a one at column $\jmin$. 
By (S1) and (S4), $r(i) \neq i'$ for all $i, i' \in [t]$ and $r(i) \neq r(i')$ 
whenever $i \neq i'$.
\begin{figure}[H]
\centering
\scalebox{1} 
{
\begin{pspicture}(0,-1.7304556)(8.025312,1.7304556)
\psline[linewidth=0.02cm](1.8778125,1.5099624)(3.0778124,1.5099624)
\psline[linewidth=0.02cm](3.4578125,1.5099624)(4.6978126,1.5099624)
\usefont{T1}{ptm}{m}{n}
\rput(3.3092186,1.4999624){$1$}
\usefont{T1}{ptm}{m}{n}
\rput(4.909219,1.4999624){$0$}
\usefont{T1}{ptm}{m}{n}
\rput(0.9476563,1.4999624){row $1$}
\psline[linewidth=0.02cm](5.0778127,1.5099624)(6.3178124,1.5099624)
\psline[linewidth=0.02cm](1.8778125,1.0899625)(3.0778124,1.0899625)
\psline[linewidth=0.02cm](3.4578125,1.0899625)(4.6978126,1.0899625)
\usefont{T1}{ptm}{m}{n}
\rput(3.3092186,1.0799625){$1$}
\usefont{T1}{ptm}{m}{n}
\rput(4.909219,1.0799625){$0$}
\usefont{T1}{ptm}{m}{n}
\rput(0.9476563,1.0799625){row $2$}
\psline[linewidth=0.02cm](5.0778127,1.0899625)(6.3178124,1.0899625)
\psline[linewidth=0.02cm](1.8778125,0.28996247)(3.0778124,0.28996247)
\psline[linewidth=0.02cm](3.4578125,0.28996247)(4.6978126,0.28996247)
\usefont{T1}{ptm}{m}{n}
\rput(3.3092186,0.27996245){$1$}
\usefont{T1}{ptm}{m}{n}
\rput(4.909219,0.27996245){$0$}
\usefont{T1}{ptm}{m}{n}
\rput(0.96765625,0.27996245){row $t$}
\psline[linewidth=0.02cm](5.0778127,0.28996247)(6.3178124,0.28996247)
\psline[linewidth=0.02cm](1.8778125,-0.29003754)(3.0778124,-0.29003754)
\psline[linewidth=0.02cm](3.4578125,-0.29003754)(4.6978126,-0.29003754)
\usefont{T1}{ptm}{m}{n}
\rput(3.3092186,-0.30003753){$0$}
\usefont{T1}{ptm}{m}{n}
\rput(4.909219,-0.30003753){$1$}
\usefont{T1}{ptm}{m}{n}
\rput(0.76765627,-0.30003753){row $r(1)$}
\psline[linewidth=0.02cm](5.0778127,-0.29003754)(6.3178124,-0.29003754)
\psline[linewidth=0.02cm](1.8778125,-0.7100375)(3.0778124,-0.7100375)
\psline[linewidth=0.02cm](3.4578125,-0.7100375)(4.6978126,-0.7100375)
\usefont{T1}{ptm}{m}{n}
\rput(3.3092186,-0.7200375){$0$}
\usefont{T1}{ptm}{m}{n}
\rput(4.909219,-0.7200375){$1$}
\usefont{T1}{ptm}{m}{n}
\rput(0.76765627,-0.7200375){row $r(2)$}
\psline[linewidth=0.02cm](5.0778127,-0.7100375)(6.3178124,-0.7100375)
\psline[linewidth=0.02cm](1.8778125,-1.4900376)(3.0778124,-1.4900376)
\psline[linewidth=0.02cm](3.4578125,-1.4900376)(4.6978126,-1.4900376)
\usefont{T1}{ptm}{m}{n}
\rput(3.3092186,-1.5000376){$0$}
\usefont{T1}{ptm}{m}{n}
\rput(4.909219,-1.5000376){$1$}
\usefont{T1}{ptm}{m}{n}
\rput(0.76765627,-1.5000376){row $r(t)$}
\psline[linewidth=0.02cm](5.0778127,-1.4900376)(6.3178124,-1.4900376)
\psbezier[linewidth=0.04](6.6986747,0.9254529)(6.417402,0.94966185)(6.695133,1.7104558)(6.433669,1.6419834)
\psbezier[linewidth=0.04](6.6986465,0.931634)(6.4079576,0.89401394)(6.7023554,0.10954438)(6.4498897,0.1956179)
\psbezier[linewidth=0.04](6.6986747,-0.8945471)(6.417402,-0.8703382)(6.695133,-0.109544225)(6.433669,-0.1780166)
\psbezier[linewidth=0.04](6.6986465,-0.888366)(6.4079576,-0.92598605)(6.7023554,-1.7104557)(6.4498897,-1.6243821)
\usefont{T1}{ptm}{m}{n}
\rput(7.352031,0.93996245){$t$ rows}
\usefont{T1}{ptm}{m}{n}
\rput(7.352031,-0.9000375){$t$ rows}
\psdots[dotsize=0.08](3.2978125,0.8299625)
\psdots[dotsize=0.08](3.2978125,0.68996245)
\psdots[dotsize=0.08](3.2978125,0.5299625)
\psdots[dotsize=0.08](4.9178123,0.8299625)
\psdots[dotsize=0.08](4.9178123,0.68996245)
\psdots[dotsize=0.08](4.9178123,0.5299625)
\psdots[dotsize=0.08](3.3178124,-0.95003754)
\psdots[dotsize=0.08](3.3178124,-1.0900376)
\psdots[dotsize=0.08](3.3178124,-1.2500376)
\psdots[dotsize=0.08](4.9178123,-0.95003754)
\psdots[dotsize=0.08](4.9178123,-1.0900376)
\psdots[dotsize=0.08](4.9178123,-1.2500376)
\end{pspicture} 
}

\end{figure}
Along the rows in the set $[t] \cup \{r(i): i \in [t]\}$, the weights of the two columns
$\jmax$ and $\jmin$ are the same (equal to $t$). 
The other rows of $\bM$, because of (\ref{eq:19}), must contribute at least as much 
to the weight of column $\jmin$ as to the weight of column $\jmax$. 
Therefore, in total, the weight of column $\jmax$ is not larger than the weight
of column $\jmin$ of $\bM$. This conclusion contradicts the fact that 
$\max \geq \min + 2$. 
\end{proof} 
\vskip 5pt

We now discuss the complexity of Algorithm~1. 
In the initial matrix $\widetilde{\bM}$, the difference between the maximum
and the minimum column weights is at most $k - 1$. 
Therefore, according to the proof of Lemma~\ref{lem:correctness}, the repeat loop
finishes after at most $(k-1)\lfloor \frac{n}{2}\rfloor$ iterations. 
It is obvious that all steps in each iteration can be done in polynomial time in $n$ 
and $k$, except for Step~7. It is not straightforward that the verification of (P3) for a 
given $k \times n$ matrix can be done in polynomial time. However, it can be 
shown that by considering a special one-source $k$-sink network 
(of linear size in $k$ and $n$) associated 
with each matrix, (P3) is equivalent to the condition that in this network, 
the minimum capacity of a cut between the source and any sink is at least $n$. 
On any network, this condition can be verified in polynomial time using 
the famous network flow algorithm (see, for instance~\cite{AhujaMagnantiOrlin}).  
Therefore, Algorithm~1 runs in polynomial time in $k$ and $n$. 
We omit the proof due to lack of space. 
Interested reader can find the proof online at \cite{supplement}. 

\section{Acknowledgment}
The first author thanks Yeow Meng Chee for informing him of
the Gale-Ryser Theorem. 
\bibliographystyle{IEEEtran}
\bibliography{Sparse-Balanced-GM-MDS}

\end{document}